\documentclass[letterpaper]{article} % DO NOT CHANGE THIS
\usepackage{aaai23}  % DO NOT CHANGE THIS
\usepackage{times}  % DO NOT CHANGE THIS
\usepackage{helvet}  % DO NOT CHANGE THIS
\usepackage{courier}  % DO NOT CHANGE THIS
\usepackage[hyphens]{url}  % DO NOT CHANGE THIS
\usepackage{graphicx} % DO NOT CHANGE THIS
\usepackage{natbib}  % DO NOT CHANGE THIS AND DO NOT ADD ANY OPTIONS TO IT
\usepackage{caption} % DO NOT CHANGE THIS AND DO NOT ADD ANY OPTIONS TO IT
\usepackage{times}

\usepackage{soul}
\usepackage{url}
\usepackage[utf8]{inputenc}
\usepackage{caption}
\usepackage{graphicx}
\usepackage{amsmath}
\usepackage{booktabs}
\usepackage{mathrsfs}
\usepackage{setspace}
\usepackage{amssymb}
\usepackage{amsthm}
\usepackage{enumitem}
\usepackage{subcaption}
\usepackage{color}
\usepackage{subcaption}
\usepackage{cleveref} 
\usepackage{booktabs}
\usepackage{multirow}
\usepackage[linesnumbered,ruled,vlined]{algorithm2e}
\usepackage{algpseudocode}
\newtheorem{thm}{Theorem}
\newtheorem{lem}{Lemma}

\newtheorem{exmp}{Example}

\newtheorem{mec}{Mechanism}
\newtheorem{coro}{Corollary}
\newtheorem{defi}{Definition}

\title{Utility Maximizer or Value Maximizer: \\ Mechanism Design for Mixed Bidders in Online Advertising}
\author{
    Hongtao Lv,\textsuperscript{\rm 1,2} 
    Zhilin Zhang,\textsuperscript{\rm 3} 
    Zhenzhe Zheng,\textsuperscript{\rm 2}\thanks{Z. Zheng is the corresponding author.}
    Jinghan Liu,\textsuperscript{\rm 4} \\
    Chuan Yu,\textsuperscript{\rm 3} 
    Lei Liu,\textsuperscript{\rm 1} 
    Lizhen Cui,\textsuperscript{\rm 1} 
    and Fan Wu\textsuperscript{\rm 2}
}
\affiliations{
    \textsuperscript{\rm 1} School of Software \& Joint SDU-NTU Centre for Artificial Intelligence Research (C-FAIR), Shandong University, China\\
    \textsuperscript{\rm 2} Department of Computer Science and Engineering, Shanghai Jiao Tong University, China \\
    \textsuperscript{\rm 3} Alibaba Group, China
    \textsuperscript{\rm 4} SJTU-ParisTech Elite Institute of Technology, Shanghai Jiao Tong University, China\\
    \{lht, l.liu, clz\}@sdu.edu.cn, \{zhangzhilin.pt, yuchuan.yc\}@alibaba-inc.com, \{zhengzhenzhe, u\_nivers, wu-fan\}@sjtu.edu.cn
}

\usepackage{bibentry}
\begin{document}

\maketitle

\begin{abstract}
Digital advertising constitutes one of the main revenue sources for online platforms. In recent years, some advertisers tend to adopt auto-bidding tools to facilitate advertising performance optimization, making the classical \emph{utility maximizer} model in auction theory not fit well. Some recent studies proposed a new model, called \emph{value maximizer}, for auto-bidding advertisers with return-on-investment (ROI) constraints. However, the model of either utility maximizer or value maximizer could only characterize partial advertisers in real-world advertising platforms. In a mixed environment where utility maximizers and value maximizers coexist, the truthful ad auction design would be challenging since bidders could manipulate both their values and affiliated classes, leading to a multi-parameter mechanism design problem. In this work, we address this issue by proposing a payment rule which combines the corresponding ones in classical VCG and GSP mechanisms in a novel way. Based on this payment rule, we propose a truthful auction mechanism with an approximation ratio of $2$ on social welfare, which is close to the lower bound of at least $\frac{5}{4}$ that we also prove. The designed auction mechanism is a generalization of VCG for utility maximizers and GSP for value maximizers. 
\end{abstract}

\section{Introduction}

Digital advertising is one of the most successful applications of auction theory, and it serves as a primary source of revenue for online platforms, such as Google, Facebook, Alibaba, and Baidu. In a typical scenario of selling advertising slots, the online platforms conduct ad allocation and compute corresponding payments for advertisers by  Vickrey-Clarke-Grove (VCG) mechanism~\cite{vickrey1961counterspeculation,clarke1971multipart,groves1973incentives} or generalized second price (GSP) auction~\cite{varian2007position,edelman2007internet}. It is widely known that VCG is a truthful mechanism, while GSP, as a more pervasive alternative in industry, is not truthful but has envy-free Nash equilibria, and all such equilibria yield no lower revenue than VCG. 
%\emph{i.e.}, the price is the minimum bid for an advertiser to maintain her slot. 
%has operated for nearly two decades.

The existing analysis on VCG and GSP mainly builds upon the quasi-linear utility model, also called as \emph{utility maximizer} (UM), \emph{i.e.}, an advertiser aims to optimize the difference between the value of the allocation and her payment. However, in modern online advertising systems, many advertisers have started to use auto-bidding tools, where they only set high-level constraints specifying a targeted return on investment (ROI) constraint (\emph{i.e.}, a targeted minimum ratio between the obtained value and the payment) under a certain budget \cite{aggarwal2019autobidding,he2021unified,yang2019bid,zhang2014optimal}.
%maximizes her obtained allocation within certain budget and/or return on investment (ROI) constraints (\emph{i.e.}, the density of allocation in payment) []. %That is, they usually take the payment as a constraint, instead of a term in the objective function. This pattern substantially deviates from the classical quasilinear utility model in auction theory.
When the ROI constraint is close to one, it could be captured by the \emph{individual rationality} property~\cite{noam2008algorithmic} in the UM model. Nevertheless, when the targeted ROI becomes large, the classical UM model could not capture the behaviors of advertisers in auto-bidding~\cite{szymanski2006impact,cavallo2017sponsored}.
As a consequence, \citeauthor{wilkens2017gsp} \shortcite{wilkens2016mechanism,wilkens2017gsp} proposed another model for advertisers, called \emph{value maximizer} (VM), which specifies the value of allocation as her objective and the payment as the second-order objective. It is empirically shown in their work that, as long as the ROI requirement of an advertiser is moderately high (above 2 or 3), her behavior pattern would be close to that in the VM model. This new model brings about an important theoretical result, \emph{i.e.}, GSP auction becomes a truthful mechanism for VMs, which offers a new perspective on the prevalence of GSP. This result renders the VM model attractive for both the academics \cite{niazadeh2022fast,cavallo2017sponsored} and the industry \cite{wang2021designing} in recent years.

However, either the model of UM or VM characterizes only a part of advertisers, \emph{i.e.}, UMs represent the advertisers with relatively low ROI constraints while VMs are those bidders with relatively high ROI constraints. As the ROI constraints of advertisers would be various in real-world advertising platforms, a natural problem arises: \emph{how to design a truthful mechanism when both classes of UMs and VMs coexist in an advertising system?}
This problem involves two aspects: 1) when the class information is public, \emph{i.e.}, an advertiser could only misreport her value, and 2) when the class information is private, \emph{i.e.}, an advertiser could misreport both her affiliated class and her value. Even in the simpler former case, it is a nontrivial mechanism design problem as the strategic behaviors of advertisers are unclear in this setting, and apparently, the VCG or GSP mechanism would not be truthful. The latter case is more practical since advertisers can easily manipulate their ROI constraints to change the corresponding class if they could benefit. This private class information introduces substantial challenges as the problem falls within the field of multi-parameter mechanism design~\cite{noam2008algorithmic}, which is hard to resolve in general.

In this work, we address the above problems by designing truthful mechanisms for both cases with public classes and private classes. %We first observe that the classic concept of \emph{social welfare} is not well-defined for value maximizers, so we borrow a concept of \emph{liquid social welfare} (LSW) as a new metric for the mixed classes of bidders from previous literature on budget-constrained bidders.
We show that the following \underline{M}echanism for mixed bidders with \underline{PU}blic classes (MPU) is truthful with respect to their values and guarantees optimal social welfare: ranking advertisers by their bids for ad allocations, and  then charging UMs following the VCG payment rule and charging VMs following the GSP payment rule. MPU is intuitive but ``unfair" for VMs to some extent since VMs may suffer a higher payment with the same allocation compared with UMs. This implies that MPU would be untruthful in the case of private classes. Therefore, we further propose a new \underline{M}echanism for mixed bidders with \underline{PR}ivate classes (MPR). The key idea of MPR is to specify a payment for each slot instead of each advertiser, which takes the maximum of the VCG-style payment induced from the closest lower UM and the GSP-style payment induced from the closest lower VM. With this payment rule, MPR conducts the ad allocation by filling the slots with all VMs in a bottom-up manner, and then iteratively assigning a slot with the best utility for a UM. We prove that this mechanism is truthful with respect to both the value and class information, and at the same time, it guarantees an approximation ratio of at most $2$ on social welfare. We also prove that no mechanisms could achieve a better approximation ratio than $\frac{5}{4}$.

The main contribution of our work lies in the following two aspects. On the one hand, we are the first to study truthful mechanism design for the hybrid setting with the coexistence of both UMs and VMs, which offers interesting insights into the combination of VCG and GSP mechanisms. Both MPU and MPR reduce to VCG when all bidders are UMs, and become GSP when all bidders are VMs. On the other hand, our work takes a substantial step towards the understanding of mechanism design for bidders with ROI constraints, which has been an emerging topic in recent years \cite{balseiro2021landscape,golrezaei2021auction,li2020incentive}. MPR implies that a truthful mechanism may allocate higher slots to bidders with high ROI constraints (\emph{i.e.}, VMs in our setting), although their bids are lower than those with low ROI constraints.

\section{The Model}

We study the standard model of an advertising auction. There are $K$ ad slots, indexed by $k \in \{1,2,...,K\}$ in a bottom-up manner. The slot $k$ has click-through-rate (CTR) $x_k$, and we assume that $x_K\ge x_{K-1}\ge...\ge x_{1}>0$. We also use $k=0$ with $x_{0}=0$ to denote a dummy slot below the lowest one.
There is a set of bidders $\mathcal{N} = \{1,2,...,n\}$, indexed by $i$, each of which has a private value $v_i$ for a click. Without loss of generality, we assume $n>K$ and $v_i\ne v_j$ for each pair of bidders $i$ and $j$ for ease of presentation.
An auction mechanism $M$ picks an allocation outcome $\Pi=\{\pi_1, \pi_2,...,\pi_K \}$ and charges the price $p_i$ of a click for each bidder $i$, where $\pi_k$ denotes the bidder whose ad is assigned the slot $k$. We also use $a_i$ to denote the index of the slot allocated to bidder $i$, \emph{i.e.}, $a_{i} = k$ if $\pi_{k} = i$. 

We consider the setting where each bidder could be either a \emph{utility maximizer} (UM) or a \emph{value maximizer} (VM). Their definitions are described as follows:
\begin{defi}[Utility Maximizer, UM]
A utility maximizer $i$ strategizes to maximize her utility $u_i = v_i x_{a_i} - p_i x_{a_i}$.
\end{defi}
\begin{defi}[Value Maximizer, VM \cite{wilkens2017gsp}]
A value maximizer $i$ strategizes to maximize her objective $u_i = v_i x_{a_i}$ while keeping the payment $p_i\le v_i$; among outcomes with equal objective, a lower price is preferred.
\end{defi}
\noindent
Intuitively, a UM follows the standard model of bidders in online advertising, while a VM prioritizes her obtained allocation over her payment. For ease of presentation, we also call the objective $u_i$ of a VM as her ``utility".
We use $\tau_i \in \{UM, VM\}$ to denote the class of bidder $i$, and $\theta_i = (v_i, \tau_i)$ as her \emph{type}\footnote{We distinguish the terms ``class" and ``type" in this work to facilitate exposition.}. We assume a bidder could misreport both her value and her class, \emph{i.e.}, a bidder may report her type as $\hat{\theta}_i =\{\hat{v}_i, \hat{\tau}_i\}$ with $\hat{v}_i\ne v_i$ and/or $\hat{\tau}_i\ne \tau_i$. 
Furthermore, we use $\theta$ to denote the type profile of all bidders, and $\theta_{-i}$ as that of all bidders except $i$. With these notations, we define $p_i(\hat{\theta}_i|\theta_i,\theta_{-i})$ as the payment of bidder $i$ by reporting her type as $\hat{\theta}_i$ under her true type $\theta_i$ and the type profile of others $\theta_{-i}$, and $u_i(\hat{\theta}_i|\theta_i,\theta_{-i})$ as the corresponding utility.

Two fundamental desiderata in mechanism design are \emph{incentive compatibility} (IC) and \emph{individual rationality} (IR).
\begin{defi}[Incentive Compatibility, IC]
A mechanism is incentive compatible if and only if 
$$u_i(\theta_i|\theta_i,\theta_{-i})\ge u_i(\hat{\theta}_i|\theta_i,\theta_{-i}), \quad \forall \hat{\theta}_i\ne \theta_i, \theta_{-i},  i\in \mathcal{N}.$$
\end{defi}
\begin{defi}[Individual Rationality, IR]
A mechanism is individually rational if and only if
$$p_i(\theta_i|\theta_i,\theta_{-i})\le v_i,\quad \forall \theta_{-i}, i\in \mathcal{N}.$$
%$$u_i(\theta_i|\theta_i,\theta_{-i})\ge 0,\quad \forall \theta_{-i}, i\in \mathcal{N}.$$
\end{defi}

\noindent
In other words, IC guarantees that all bidders in the mechanism do not have incentives to misreport their types, IR guarantees that the bidders would never suffer a negative utility when truthfully bidding. We note that 
we will loosely use the term ``truthful" to describe a mechanism that is both IC and IR. It is well-known that VCG is truthful for UMs. In recent years, it is also proved that GSP is truthful for VMs \cite{wilkens2017gsp}. %These studies focused on settings where there is only one class of bidders, and the class information is naturally public.
In this work, our goal is to design a truthful mechanism for mixed bidders while maximizing the overall welfare. Here, the mixed bidders could be either UMs or VMs. 

It is worth noting that the classical concept of \emph{social welfare}, \emph{i.e.}, the sum of the utilities of all bidders and the revenue of the seller, is not well-defined for VMs, because the transaction amount between bidders and the seller could not be canceled in the social welfare calculation, different from that of UMs.
Therefore, we borrow a concept of \emph{liquid social welfare} (LSW) from previous literature on mechanism design for budget-constrained bidders \cite{dobzinski2014efficiency,deng2021towards}.

\begin{defi}[Liquid Social Welfare, LSW]
The liquid social welfare of an allocation outcome $\Pi$ in a mechanism is the sum of the maximum willingness-to-pay of all bidders for the allocation, \emph{i.e.}, $$Wel(\Pi) = \sum_{k=1}^K v_{\pi_k} x_{k}.$$
\end{defi}
\noindent
With the definition of LSW, we can easily obtain that VCG is LSW-optimal for UMs and GSP is LSW-optimal for VMs. Therefore, when we strive to design an LSW-optimal mechanism for mixed bidders, a natural requirement is \emph{Robustness}:
\begin{defi}[Robustness]
A mechanism is robust if the outcome is the same as VCG when all bidders are UMs, and it is the same as GSP when all bidders are VMs.
\end{defi}
\noindent
Robustness guarantees that the mechanism is a natural generalization from existing mechanisms. 

In summary, a desired auction mechanism should be IC, IR, robust, and LSW-optimal.
Moreover, since we focus on truthful mechanisms, we do not distinguish $\theta_i$ and $\hat{\theta}_i$ when clear from the context.

\section{Warming Up: Public Classes}
To begin with, we study a basic setting where the classes of bidders are public. In such case, the problem falls within the field of single-parameter mechanism design \cite{noam2008algorithmic}, which would be simpler but useful for theoretical understanding. We propose the following Mechanism for bidders with PUblic classes (MPU):

\begin{mec} [MPU]
\quad
\label{mechanism1}
\begin{itemize}
    \item \textbf{Allocation}: Ranking the bidders by their values, and allocating slots accordingly from top to bottom.% \emph{i.e.}, for each slot $0\le k\le K$, $\pi_k$ is the bidder ranked as the $(K+1-k)$th highest one. 
    \item \textbf{Payment}: For each bidder $i$ with $1\le a_i\le K$, let $j$ be the bidder in the next lower slot, \emph{i.e.}, $a_j=a_i-1$.
    \begin{itemize}
        \item If $i$ is a UM, then $p_i = \frac{1}{x_{a_i}} \sum_{k=0}^{a_j} v_{\pi_k}(x_{k+1}-x_{k});$
        \item If $i$ is a VM, then $p_i = v_{j}.$
    \end{itemize}
\end{itemize}
\end{mec}

Intuitively, MPU directly allocates the slots to bidders by their values, regardless of their classes. Furthermore, the payment of UMs follows VCG, and that of VMs follows GSP (recall that we use bottom-up indexes for slots). We next show that this mechanism is IC, IR, and robust, which also guarantees the optimal LSW.

\begin{thm}
\label{thm1}
When the classes of bidders are public information, MPU is IC, IR, robust, and LSW-optimal.
\end{thm}
Following the previous works on UM \cite{noam2008algorithmic} and VM \cite{wilkens2017gsp}, the proof of Theorem \ref{thm1} is straightforward, so we omit it here.

\section{Private Classes}

In the preceding section, we have developed the optimal mechanism for mixed bidders with public classes. However, when the class information is private, the problem turns out to be a multi-parameter mechanism design problem, which is hard to resolve in general. 
We can first examine whether MPU is IC for the setting of private classes. It could be easily observed that, if a VM misreports her class as UM while truthfully reporting her value, she may enjoy a lower payment without changing her allocation. In other words, MPU is ``unfair" to VMs in some sense, and this unfairness may bring the probability of strategic manipulation when the class information is private. 

\subsection{Mechanism for Mixed Bidders with Private Classes}

The above analysis implies that, in a truthful mechanism for the setting of private classes, the payment of a bidder should rely only on her allocated slot, rather than her class. In other words, suppose a bidder is allocated an identical slot in two cases while the allocation outcomes for others are also the same, then no matter she is a UM or VM, the payment should be the same. This requirement leads us to devise a payment rule for slots instead of for bidders, and this rule should combine VCG and GSP based on the types of bidders below the slot. Therefore, we propose the price of a slot as the maximum of the following two terms: 1) the VCG-style payment derived from the closest lower UM, and 2) the GSP-style payment derived from the closest lower VM. Specifically, for a slot $k\ge 1$, let the closest VM below $k$ be $i_{V}$, located at $k_{V}$, and the closest lower UM be $i_{U}$, located at $k_{U}$, then the price of  slot $k$ is given as
\begin{equation}
\label{eq15}
    p^{(k)} = \max\{\hat{p}^{(k)}_{U}, v_{i_{V}}\},
\end{equation}
where
\begin{equation}
\label{eq16}
    \hat{p}^{(k)}_{U} = \frac{1}{x_k}(p^{(k_{U})}x_{k_{U}} + v_{i_{U}}(x_{k}-x_{k_{U}})).
\end{equation}
When there is no VM or UM below slot $k$, we assign the corresponding payment term as 0.
Given this payment rule, we propose a Mechanism for mixed bidders with PRivate Classes (MPR) which is IC, IR, and robust, while achieving a desired approximation ratio in terms of LSW.

\begin{algorithm}[t]
\caption{MPR}
\label{alg1}
\KwIn{The type profile $\theta$ of all bidders, the CTR $x_k$ of all slots. }
\KwOut{The allocation and payment outcome for each bidder $i$.}
$p_{i} \leftarrow 0, \forall i\in \{1,...,n\}$;

Sort all bidders by their values;

Let $N$ be the set of top $K$ bidders by their values, and $i$ be the $(K+1)$th highest one;

$\pi_{0} \leftarrow i$;

$S\leftarrow$ the set of all UMs in $N$;

$T\leftarrow$ the set of all VMs in $N$;

\tcp{Allocate slots to all VMs in $T$.}

\If{$|T|> 0$ }{

\For{$k$ from $1$ to $|T|$}{
$\pi_k \leftarrow$ the VM with the $k$th lowest value in $T$;
%$\pi_k \leftarrow \arg\min_{i\in T'} \{v_i\}$.
%$p^{(k)} \leftarrow v_{\pi_{k-1}}$.
%$p^{(k+1)} \leftarrow \frac{1}{x_{k+1}} \sum_{k'=0}^{k} v_{\pi_{k'}}(x_{k'+1}-x_{k'})$.
%$T'\leftarrow T' - \{\pi_k\}$.
}}

Update the payment $p^{(k)}$ for slots $1\le k \le |T|+1$ by equation (\ref{eq15});

\tcp{Allocate slots to UMs iteratively.}
\While{$|S|>0$}{

$i\leftarrow \arg\min_{j\in S} \{v_j\}$;

$\bar{k} \leftarrow K-|S|+1$; %\resizebox{0.255\textwidth}{!}{\tcp{The top available slot.}}

$k^{i}\leftarrow \arg\max_{1\le k\le \bar{k}} \{x_{k}(v_i- p^{(k)})\}$;

\If{$k^{i} \ne \bar{k}$}{
\resizebox{0.38\textwidth}{!}{\tcp{Move existing bidders up one slot.}}
\For{$k$ from $\bar{k}$ down to $k^{i}+1$}{
$\pi_{k} \leftarrow \pi_{k-1}$;
}
}
$\pi_{k^{i}} \leftarrow i$;

Update the payment $p^{(k)}$ for slots $k^{i}+1 \le k \le \bar{k}+1$ by equation (\ref{eq15});

$S\leftarrow S\backslash \{i\}$;
}

$p_{\pi_k} \leftarrow p^{(k)}, \forall k\in \{1,...,K\}$;

Return $\pi_k$ and $p_i$ for each slot $k$ and each bidder $i$.
\end{algorithm}

The complete pseudo-code of MPR is presented in Algorithm \ref{alg1}. The key idea behind MPR is to fill the slots with all VMs, and then iteratively assign a slot with the best utility for a UM. 
In Algorithm \ref{alg1}, MPR first sorts the bidders by their values and obtains the set of the top $K$ bidders as $N$ (Lines 2-3). The bidder with the $(K+1)$th highest value would be allocated the dummy slot indexed by 0 as a basis for pricing (Line 4). We denote $S$ as all the UMs in $N$, and $T$ as all the VMs correspondingly (Lines 5-6). Next, MPR fills the lowest slots with all the VMs in $T$ according to their values (Lines 7-9). Then we compute the prices for slots $1\le k\le |T|+1$ following (\ref{eq15}) based on the allocated VMs (Line 10). Indeed, since no UMs in $S$ are allocated slots at this step, the term of $\hat{p}_U^{(k)}$ is set to $0$, hence the price computation could be simplified as the GSP payment rule. After calculating the prices for each slot, we determine the allocation for UMs. The UM not yet assigned with the lowest value is picked as $i$, and we choose the optimal slot $k^i$ for her, \emph{i.e.}, the slot with the highest utility (if there is a tie, choose the lowest one) (Lines 12-14). It is noteworthy that there are two kinds of choices of slots: 1) $k^{i}=\bar{k}=K-|S|+1$, \emph{i.e.}, the slot above all assigned bidders (note that $K=|T|+|S|$ in the first round and $S$ will be updated as the set of all unassigned UMs during the process); and 2) $k^{i}<\bar{k}$, \emph{i.e.}, a slot which an existing bidder occupies. In the former choice, we only need to allocate the slot $k^i$ to $i$. In the latter choice, we first move all bidders at and above slot $k^i$ up one slot and then allocate the slot $k^i$ to $i$ (Lines 15-18). Next, the prices of slots above $k^i$ are updated based on the value of the newly inserted UM (Line 19). Then the UM $i$ is removed from $S$, and the process is repeated until all UMs in $S$ are assigned a slot (Line 20). Finally, if a bidder is assigned the slot $k$, her payment for a click would be the price of the slot; otherwise, her payment would be 0 (Line 21).

Next, we provide an example to illustrate the running process of MPR. One can observe from the example that, interestingly, VMs with lower values may obtain a higher slot than UMs with higher values.

\begin{exmp}
\label{exmp1}
Assume there are four slots and five bidders, and their CTRs or types are presented in Fig. \ref{fig:example}. In MPR, we first place the bidder with the fifth highest value, \emph{i.e.}, bidder A, at the dummy slot $0$. Then the remaining VMs B and C are placed at slots $1$ and $2$, respectively. We can hence calculate the prices for slots 1, $2$ and $3$, \emph{i.e.}, $p^{(1)}=6, p^{(2)}=7, p^{(3)}=8$. With these prices, we get the utilities of bidder D at each slot: $0.3$, $0.4$, and $0.3$ for slots $1$, $2$, and $3$, respectively. Then the optimal one, \emph{i.e.}, slot $2$ is allocated to her, and bidder C is moved to slot $3$. Next, the prices of slots $3$ and $4$ are updated as $p^{(3)}=\frac{23}{3}, p^{(4)}=8$ by (\ref{eq15}). Finally, we get the utilities of bidder E at each slot: $0.4$, $0.6$, $0.7$, and $0.8$ for slots $1$, $2$, $3$, and $4$, respectively, and thus, slot $4$ is allocated to bidder E. The payments of bidders for a click are given by the corresponding prices of their obtained slots. 
\end{exmp}

\begin{figure}
    \centering
    \includegraphics[width = .3\textwidth]{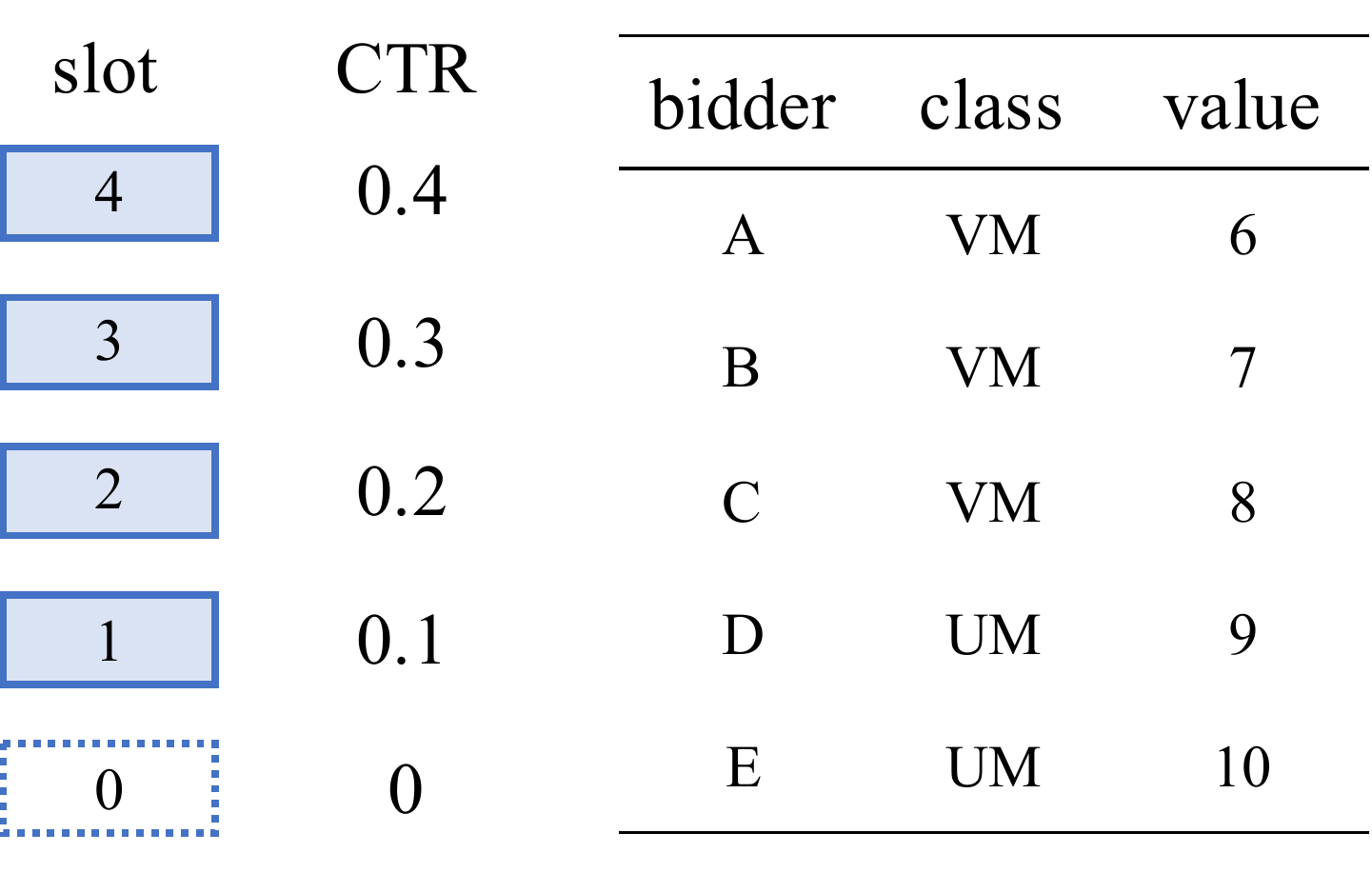}
    \caption{The illustration of four slots with their CTRs, and five bidders with their classes and values in Example \ref{exmp1}.}
    \label{fig:example}
\end{figure}

\subsection{Game Theoretical Properties}
Before proving game theoretical properties of MPR, we first define a concept of \emph{marginal payment increase}~\cite{bachrach2016mechanism}, %as also used in \cite{bachrach2016mechanism}, 
to measure the cost performance for a bidder to obtain a higher slot. Based on this concept, we provide several lemmas to help understand the ideas behind MPR, which are helpful to the proof of IC and IR.

\begin{defi}
For two slots $k'>k$, the marginal increase of payment is defined as 
\begin{equation}
\Delta(k,k') \triangleq \frac{p_{\pi_{k'}}x_{k'} - p_{\pi_{k}}x_{k}}{x_{k'}-x_k}.
\end{equation}
\end{defi}

% Due to the space limitation, the proofs of the following lemmas are deferred to the supplementary file.
\begin{lem}
\label{marginal}
If a UM $i$ is allocated a slot $a_i = k$, we have $\Delta(k,k+1)=v_i$.
\end{lem}
\begin{proof}
For slot $k+1$, UM $i$ is the closest lower UM, and we denote $\hat{i}$ as the closest lower VM at slot $\hat{k}$ (if there does not exist such VM, we can directly observe that $\Delta(k,k+1) = v_i$). 
As $i$ is a UM, her utility would always be non-negative, \emph{i.e.}, $p^{(k)}\le v_i$, otherwise, she would be assigned at least the slot $k=1$, yielding a non-negative utility. Since $\hat{i}$ is also the closest lower VM for $i$, we can further derive that $v_{\hat{i}}\le p^{(k)}\le v_i$. Therefore, in the maximum function of computing $p^{(k+1)}$ by (\ref{eq15}), we have that the first term (the VCG-style payment from $i$) is always no less than the second term (the GSP-style payment from $\hat{i}$). This implies that  $\Delta(k,k+1) = v_i$.
\end{proof}

\begin{lem}
\label{same}
For two bidders $i$ and $j$, if $v_i>v_j$ and $\tau_i = \tau_j$, \emph{i.e.}, they are both UMs or both VMs, then we have $a_i>a_j$.
\end{lem} 

\begin{proof}
For VMs, it is straightforward that bidder $i$ is always allocated a higher slot than $j$ during the algorithm process. For UMs, without loss of generality, we first assume that $i$ is the UM with the lowest slot such that $a_i<a_j$ for a UM $j$ with $v_i>v_j$. Then we have that the allocation below $a_i$ would be the same in the rounds of assigning $i$ and assigning $j$. Therefore, if $a_i < a_j -1$, we know that the bidder $j$ has faced the same choice of slot $a_i$ with the same price, but she chose $a_j$. As $v_i>v_j$ and $x_{a_j} > x_{a_i}$, by the utility function of UMs, we have bidder $i$ too prefers slot $a_j$ to $a_i$, leading to a contradiction. If $a_i = a_j -1$,  by lemma \ref{marginal}, we obtain that bidder $i$ prefers slot $a_j+1$ to $a_j$ as the marginal price $\Delta(a_j, a_j+1)=v_j<v_i$, which is a contradiction.
%one can see that, when we choose the slot with the best benefit for $i$, bidder $j$ has always been allocated to a slot, named $k^j$. Then,  bidder $i$ would at least be allocated the slot $k^j+1$ as its price is $v_j$, lower than $v_i$. In a nutshell, we have $a_i>k^j+1>k^j=a_j$, which concludes the proof.
\end{proof}

\begin{lem}
\label{between}
For two bidders $i$ and $j$, if bidder $i$ is a VM and $j$ is a UM, and $a_i<a_j$, then we have $v_i<v_j$.
\end{lem}

\begin{proof}
Armed with Lemma \ref{same}, it suffices to prove the statement for neighboring slots, \emph{i.e.}, let $a_i=k$, then $a_j=k+1$. Assume $v_i>v_j$ for contradiction, then the price for bidder $j$ is at least $v_i$, \emph{i.e.}, $j$ suffers a negative utility. However, she can at least choose the lowest slot with non-negative utility, as discussed earlier, which is a contradiction.
\end{proof}
\begin{lem}
\label{marginal2}
If a UM $i$ is allocated a slot $a_i = k$, then we have $\Delta(k,k')\ge v_i, \forall k'>k.$
\end{lem}
\begin{proof}
For $k'=k+1$, we have proved the statement in Lemma \ref{marginal}. 
Next, for $k'>k+1$, we consider two cases: 1) there does not exist a UM between slot $k$ and $k'$; 2) there exists at least one UM between slot $k$ and $k'$. In the former case, let $j = \pi_{k'-1}$, then we know that $j$ is a VM, and the closest lower UM is $i$, so we get that 
\begin{equation}
\label{payment}
    p_{\pi_{k'}} = \max \{\hat{p}_{U}^{(k')}, v_{j}\} \ge \hat{p}_{U}^{(k')},
\end{equation}
where
\begin{equation}
\label{payment_UM}
    \hat{p}_{U}^{(k')} = \frac{1}{x_{k'}}(p_{\pi_k}x_{k} + v_i(x_{k'}-x_{k})).
\end{equation}
Combining (\ref{payment}) and (\ref{payment_UM}), we can derive that $\Delta(k,k')\ge v_i$.
In the latter case, \emph{i.e.}, there exists a set of UMs $U$ between $k$ and $k'$. By Lemma \ref{same}, we have the value of any UM in the set is higher than $v_i$. Then, for any pair of slots of neighboring UMs $k^+, k^-$, such that $\pi_{k^+}, \pi_{k^-} \in U\cup \{i\}$ with $k^+> k^-$, we can obtain $\Delta(k^-,k^+)\ge v_{k^-}\ge v_i$, given the above analysis of the the former case. Also, let $k''$ be the highest slot of UMs in $U$, we have $\Delta(k'', k')\ge v_{k''}> v_i$. Finally, the marginal increase of payment between $k$ and $k'$ would be a linear combination of the ones for each neighboring pair of bidders, that is, $\Delta(k, k')\ge v_i$.
\end{proof}

\begin{coro}
MPR is a robust mechanism.
\end{coro}
This corollary could be derived directly from the algorithm process of MPR. It implies that MPR is a good generalization and combination of VCG and GSP, which matches our intuition well.
Building upon the above results, we next prove that MPR is IR and IC.

\begin{thm}
MPR is individually rational.
\end{thm}

\begin{proof}
For UMs, the proof is straightforward since a UM could at least choose the lowest slot, achieving a non-negative utility. For VMs, by Lemma \ref{same}, we first have that the value of all other VMs below a VM $i$ should always be lower than her, hence the price induced from the closest lower VM is lower than $v_i$, and it suffices to analyze the price induced from the closest lower UM. We denote this UM as $j$ (if such UM does not exist, then the theorem trivially holds). Since $j$ chooses the slot $a_j$ instead of the slot $a_i$ when allocating the slot to her (note that $i$ is located at slot $a_i-1$ in this round), we can derive that 
\begin{equation}
\label{eq3}
    x_{a_j}(v_j-p^{(a_j)})\ge x_{a_i}(v_j-p^{(a_i)}),
\end{equation}
where $p^{(a_j)}$ and $p^{(a_i)}$ are the prices of slot $a_j$ and $a_i$ at the round of allocating bidder $j$, respectively. Furthermore, by the payment rule (\ref{eq15}), we know that there are two possibilities: 1) $p^{(a_i)}= v_i$; 2) $p^{(a_i)}= \frac{1}{x_{a_i}}(p^{(a_{\hat{j}})}x_{a_{\hat{j}}} + v_{\hat{j}}(x_{a_i}-x_{a_{\hat{j}}}))$ where $\hat{j}$ is the UM immediately below $j$. Since $v_{\hat{j}}< v_j$, we can derive that bidder $j$ would prefer slot $a_i$ than $a_j$ if the latter probability occurs. Therefore, we obtain
\begin{equation}
\label{eq4}
    p^{(a_i)} = v_i.
\end{equation}
Combining (\ref{eq3}) and (\ref{eq4}), we have that
\begin{equation}
\label{eq5}
    v_i\ge \frac{1}{x_{a_i}}(p^{(a_j)}x_{a_j} + v_j(x_{a_i}-x_{a_j})).
\end{equation}
Since $p^{(a_j)}$ remains the same in the final outcome as in the round of allocating $j$, (\ref{eq5}) indicates that $i$ would enjoy a price no more than her value, which concludes the proof.
\end{proof}

\begin{thm}
MPR is incentive compatible.
\end{thm}

\begin{proof}
We discuss the proof for UMs and VMs separately.

First, let bidder $i$ be a UM. We denote $j$ and $\hat{j}$ as the closest UMs below and above bidder $i$ in the truthful case, respectively (if such UMs do not exist, we can assume virtual ones at slot $0$ and $K+1$). Accordingly, we use $a_j$ and $a_{\hat{j}}$ to denote their slots in the truthful setting, and $a'_j$ and $a'_{\hat{j}}$ the slots in the setting where $i$ misreports her type. Next, no matter what class the bidder $i$ reports, we consider the following three cases: 1) $a'_i$, \emph{i.e.}, the slot of bidder $i$ when misreporting, is between $a'_j$ and $a'_{\hat{j}}$; 2) $a'_i$ is higher than $a'_{\hat{j}}$; 3) $a'_i$ is lower than $a'_{j}$. In the first case, we have that the outcome until the round of allocating bidder $j$ remains the same as the truthful setting. So, bidder $i$ faces both the choices of $a_i$ and $a'_i$ with the same price in the truthful setting, and she prefers $a_i$, implying that she would not get better off when misreporting. In the second case, by Lemma \ref{marginal2}, we have that the utility of bidder $i$ at $a'_i$ is always no better than that at slot $a'_{\hat{j}}$ by reporting her type as $(v_{\hat{j}}-\epsilon, UM)$, where $\epsilon$ is a sufficiently small positive number. Furthermore, by the analysis of the above case, we have that bidder $i$ prefers $a_i$ to the outcome of misreporting $(v_{\hat{j}}-\epsilon, UM)$. In the third case, let $\tilde{j}$ be the closest UM above $a'_i$. We can observe that the allocation of all the slots below $a'_i$ should be the same as the truthful setting. Therefore, if $a_{\tilde{j}}\ne a'_{i}$, $\tilde{j}$ have faced the same choices of $a'_i$ and $a_{\tilde{j}}$ and she prefers $a_{\tilde{j}}$ in the round of allocating $\tilde{j}$ in the truthful setting. Since $v_{i}>v_{\tilde{j}}$, we have bidder $i$ too prefers $a_{\tilde{j}}$ and could obtain it by misreporting $(v_{\tilde{j}}-\epsilon,UM)$ with sufficiently small $\epsilon$. Then by Lemma \ref{marginal}, we have bidder $i$ prefers the slot immediately above $\tilde{j}$, \emph{i.e.}, $a_{\tilde{j}}+1$, to $a_{\tilde{j}}$. One can run this process iteratively if $j>\tilde{j}$, until $a_{\tilde{j}}+1>a'_{j}$. Finally, we get that bidder $i$ prefers $a_i$ again by the analysis of the first case.

Second, let bidder $i$ be a VM. If she misreports her type and obtains a lower slot, obviously, her utility would decrease by the definition of VMs.  If she misreports her type and obtains a slot higher than a VM with a higher value, we can easily observe that the payment would be higher than $i$'s value, \emph{i.e.}, the utility of bidder $i$ would also decrease. Therefore, we only need to consider the case where bidder $i$ is allocated the same or a higher slot by misreporting her type, but not exceeding any VMs with higher values. By Lemma \ref{same}, the order of UMs, as well as the order of other VMs, will not change, so we have that all the UMs below bidder $i$ should not exceed $i$ in such case; otherwise, she could not obtain a higher or the same slot. Let $j$ be the highest UM below bidder $i$ in the truthful case, and accordingly, let $\hat{j}$ be the lowest UM above bidder $i$. Then we can derive that the allocation until the round of allocating bidder $j$ would be identical to the truthful case when $i$ misreports her type. 
%Thus, if $i$ is allocated the same slot, the payment would also be the same. So we only need to analyze what happens when bidder $i$ obtains a higher slot by manipulation. Since bidder $i$ would not report a higher value than the VMs above her, she is placed at the same place $a_i$ when allocating UM $\hat{j}$ in the untruthful setting.
Under these restrictions, we obtain that bidder $i$ must be above bidder $\hat{j}$ when she is untruthful; otherwise, she would be allocated at most the same slot $a_i$ and pay the same price. Therefore, we discuss the following two cases in the truthful setting: 1) $a_{\hat{j}} = a_i+1$, and 2) $a_{\hat{j}} > a_i+1$.

If $a_{\hat{j}} = a_i+1$, since UM $\hat{j}$ prefers the slot $a_{\hat{j}}$ in the truthful setting, we have that
\begin{equation}
\label{eq9}
    x_k (v_{\hat{j}}-p^{(k)}) < x_{a_{\hat{j}}} (v_{\hat{j}}-p^{(a_{\hat{j}})}) \le x_{a_{\hat{j}}} (v_{\hat{j}}-v_i),
\end{equation}
where $k$ is any slot below $a_{\hat{j}}$. Note that we do not need to consider the equality in the first inequality, as UMs would break ties by choosing the lower slots.
In the untruthful case, we get that bidder $i$ is allocated at least slot $a_{\hat{j}}$ and bidder $\hat{j}$ is allocated a lower slot $k$. Then by (\ref{eq15}) and (\ref{eq9}) we get
\begin{equation}
    p^{(a_{\hat{j}})}x_{a_{\hat{j}}} \ge x_k p^{(k)} + v_{\hat{j}} (x_{a_{\hat{j}}} - x_k) > v_i x_{a_{\hat{j}}},
\end{equation}
that is, bidder $i$ would suffer a negative utility at slot $a_{\hat{j}}$. Furthermore, one can easily observe that the prices of slots are non-decreasing during the algorithm process, \emph{i.e.}, the utility of bidder $i$ is still negative at higher slots. 

If $a_{\hat{j}} > a_i+1$, we have that there exist other VMs between $i$ and $\hat{j}$.
We consider the following two cases: 1) $i$ reports her class as VM, and 2) $i$ reports her class as UM. In the former case, as $i$ could only increase her value to obtain a higher slot (without exceeding higher VMs), the price of slot $a_{\hat{j}}$ is not determined by the price of VM $i$, nor do the prices of slots $k\le a_i$. Then we know that $\hat{j}$ would not choose a slot at or below $a_i$ when $i$ misreports her value. In the latter case, let $i'$ be the highest VM between $i$ and $\hat{j}$ in the truthful setting. Since $\hat{j}$ chooses the slot above $i'$ in the truthful setting but chooses a slot at or below $a_{i}$ in the untruthful setting, we can get that the price of the slot immediately above $i'$ in the round of allocating $i$ is derived from UM $\hat{j}$, making it preferred by $i$, compared with lower slots. This implies that $i$ would finally be allocated a slot above $i'$, yielding a negative utility. Therefore, we can conclude that a VM could never enjoy a higher utility by misreporting her type.
\end{proof}

\subsection{Approximation Ratio on LSW}

\begin{thm}
MPR achieves an approximation ratio of at most $2$ on LSW, \emph{i.e.}, $Wel_{MPR}\ge \frac{1}{2} Wel_{OPT}$, where $Wel_{OPT} = max_{\Pi}Wel(\Pi)$.
\end{thm}
\begin{proof}
By Lemma \ref{same} and \ref{between}, we have that the only probability for the event of $Wel_{MPR}<Wel_{OPT}$ is that, some VMs with lower values are allocated higher slots than UMs with higher values. 
We now consider each such VM $i$ in MPR in a bottom-up sequence and prove that changing the slot of $i$ with lower UMs to fit the optimal outcome leads to a difference in LSW no more than that of the welfare achieved by $i$ in MPR. Then the overall difference in LSW of MPR and the optimal outcome would be no more than that of MPR, resulting in an approximation ratio of 2.

First, let $i$ be the lowest VM in MPR who locates at a higher slot than some UMs with higher values than her.
We denote $S^{i}$ as the set of UMs below $i$ with higher values and index them by $j_c$ with $1\le c\le |S^{i}|$ in a bottom-up manner. We also use $k_c$ to denote the slot of $j_c$ in MPR. By Lemma \ref{same}, we know that there are no VMs between $i$ and UMs in $S^{i}$, otherwise, $i$ would not be the lowest one of interest. Thus we have $k_{c+1} = k_{c}+1$ for any $1\le c\le |S^{i}|-1$, and $a_i = k_{|S^{i}|}+1$. Next, we denote $\hat{p}^{S_i}$ as the payment purely induced from the UMs in $S^i$, that is
\begin{equation}
\label{eq21}
    \hat{p}^{S_i} = \frac{1}{x_{a_i}}\left(
    v_{k_{|S^{i}|}}(x_{i}-x_{k_{|S^{i}|}})+\sum_{c=1}^{|S^{i}|-1} v_{j_c}(x_{k_{c+1}}-x_{k_{c}})\right).
\end{equation}
Note that $\hat{p}^{S_i}$ would be no more than $p_i$ by equation (\ref{eq15}), then by the IR property of MPR, we have
\begin{equation}
\label{eq22}
    v_i x_{a_i}\ge p_i x_{a_i} \ge \hat{p}^{S_i} x_{a_i}.
\end{equation}
Next, we change the allocation of $i$ and UMs in $S^i$ towards the optimal outcome and denote $D^{i}$ as the difference in LSW after and before the change. Then we have that
\begin{equation}
\label{eq14}
\resizebox{0.429\textwidth}{!}{$
\begin{split}
    D^{i} &= v_{j_{|S^{i}|}}(x_{i}-x_{k_{|S^{i}|}}) + \left(\sum_{c=1}^{|S^{i}|-1} v_{j_c}(x_{k_{c+1}}-x_{k_{c}})\right) \\
    &\quad - v_{i}(x_{a_{i}}-x_{k_{1}})\\
    &\le v_i x_{a_i},
\end{split}
$}
\end{equation}
where the equality holds because after the change, all UMs in $S^i$ are moved up one slot, and $i$ is moved down from $a_i$ to $k_1$. The inequality comes from equations (\ref{eq21}) and (\ref{eq22}).

This way, we have moved $i$ and corresponding UMs towards the optimal outcome and have proved that $D^{i}\le v_i x_{a_i}$ for the lowest VM. In the next step, we set $i$ as the new lowest VM with non-empty $S^i$, and again, we denote $k_c$ as the assigned slots of UMs in the corresponding set of $S^i$ in the current allocation. Then one can observe that the newly induced $\hat{p}^{S_i}$ is no more than that induced by the allocation in MPR, and hence, all of the above analyses still hold. So we can repeat this process for all such VMs.

Next, the difference in LSW between MPR and the optimal outcome is given by
\begin{equation}
\resizebox{0.48\textwidth}{!}{$
Wel_{OPT}-Wel_{MPR} = \sum\limits_{S^i\ne \emptyset} D^i \le \sum\limits_{S^i\ne \emptyset} v_ix_{a_i}\le Wel_{MPR}.
$}
\end{equation}
The first equation comes from that we consider VMs in a bottom-up order, so that the difference induced by each VM could be computed separately. Finally, we obtain that $Wel_{MPR}\ge \frac{1}{2}Wel_{OPT}$.
\end{proof}

\begin{figure}
    \centering
    \includegraphics[width = .3\textwidth]{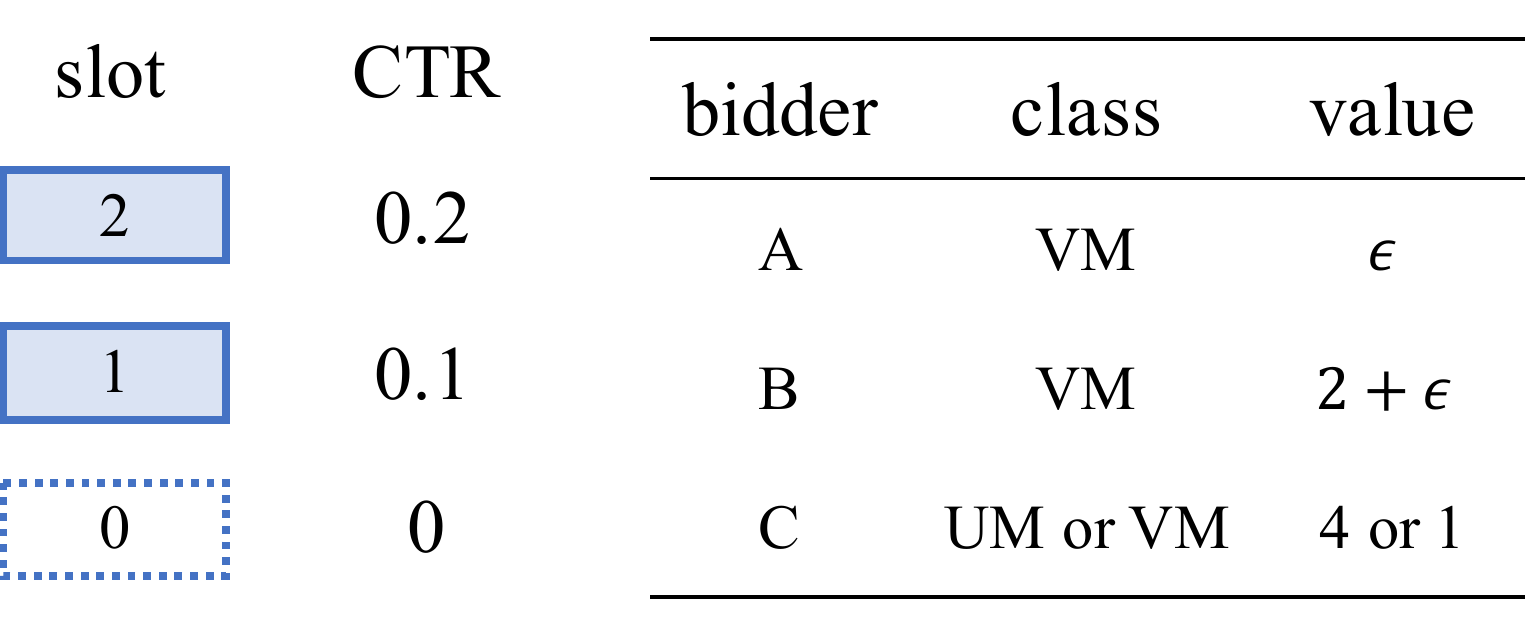}
    \caption{A counter-example for Theorem \ref{lower-bound}.}
    \label{fig:counter}
\end{figure}

\begin{thm}
\label{lower-bound}
No mechanisms that are IC, IR, and robust could guarantee an approximation ratio lower than $\frac{5}{4}$ in terms of LSW.
\end{thm}
\begin{proof}
We use a counter-example to prove this theorem. Assume an IC, IR and robust mechanism $\widetilde{M}$ achieves an approximation ratio lower than $\frac{5}{4}$ on LSW. As illustrated in Fig.~\ref{fig:counter}, let there be two slots with CTRs of $0.1$ and $0.2$ (and a dummy slot with CTR of $0$). There are three bidders A, B, and C. The types of A and B are $(\epsilon, VM)$ and $(2+\epsilon, VM)$, respectively, where $\epsilon$ is a sufficiently small positive number. We consider four cases for the type of bidder C: \textbf{Case 1)} $(4, UM)$; \textbf{Case 2)} $(4, VM)$; \textbf{Case 3)} $(1, UM)$; \textbf{Case 4)} $(1, VM)$. Since the approximation ratio on LSW is lower than $\frac{5}{4}$ in $\widetilde{M}$, the allocation outcomes are certain: in the first two cases, bidders A, B, and C get slot $0$, $1$, $2$, respectively; in the last two cases, bidders A, B, and C get slot $0$, $2$, $1$, respectively. Otherwise, one can check that any allocation outcome would result in an approximation ratio higher than $\frac{5}{4}$. Then, as discussed in previous sections, the payments of bidder C in Case 1 and Case 2 should be the same, denoted as $p_h$; otherwise, C may misreport her class. This claim also holds for Case 3 and Case 4, where the payment is denoted as $p_l$. 
Next, in Case 1, if C misreports her value as $1$, the outcome would be the same as in Case 3, and her utility should be no more than truthfully reporting $4$ by the requirement of IC. Hence we have 
\begin{equation}
    0.2(4-p_h) \ge 0.1(4-p_l),
\end{equation}
which further implies that 
\begin{equation}
\label{requirement}
    2p_h-p_l \le 4.
\end{equation}
It is noteworthy that, in Case 2 and Case 4, all bidders are VMs, thus by the requirement of robustness, we can use GSP to compute the payments, \emph{i.e.}, $p_h=2+\epsilon$ and $p_l=\epsilon$; hence we have $2p_h-p_l = 4+\epsilon > 4$, which contradicts with equation (\ref{requirement}).
This concludes our proof.
\end{proof}

\section{Related Work}
The study on VMs stems from the prosperity of auto-bidding techniques in recent years, where bidders only specify their targeted ROI and budget constraints \cite{zhang2014optimal,aggarwal2019autobidding,he2021unified}. This new pattern leads researchers to devise more practical models and mechanisms for auto-bidding advertisers. 
\citeauthor{golrezaei2021auction} \shortcite{golrezaei2021auction} and \citeauthor{balseiro2021landscape} \shortcite{balseiro2021landscape} considered \emph{ex-ante} ROI constraints while \cite{cavallo2017sponsored} considered \emph{ex-post} ROI constraints. They developed fruitful understandings of the characterizations of IC and revenue-maximizing mechanism design. However, when the allocation value and the target ROI are both private, it is hard to design optimal mechanisms due to the fundamental difficulty of multi-parameter mechanism design.
Therefore, two independent models of VMs were proposed in \cite{wilkens2017gsp} and \cite{fadaei2017truthfulness}, to characterize bidders with relatively high ROI constraints in a light way. Both of these two works considered that bidders aim to maximize their  allocation value, while the former one further assumed that a bidder set payment minimization as her second-order objective, making GSP a truthful mechanism. Our work inherits the model in the former work and extends it to a more practical  environment where VMs and UMs coexist.

% Some works proposed other models for auto-bidding advertisers by incorporating the ROI constraints.  Our work contributes to this line of studies by providing a solution in a simplified setting, where only two classes of bidders (with low and high \emph{ex-post} ROI constraints separately) exist in the system.

Another stream of research related to ours is the mixture of VCG and GSP mechanisms, where the main goal is to transition the existing GSP mechanism into VCG mechanism to adapt to modern complex advertising environments. 
The work in \cite{bachrach2016mechanism} proposed a transitional mechanism, which is similar to ours. However, they considered all bidders as traditional utility maximizers, while some are adaptive to VCG and some are non-adaptive and still use the GSP bids. Moreover, they took the classes of ``adaptive" or ``non-adaptive" bidders as public information, while we consider the private class information.
\citeauthor{hummel2018hybrid} \shortcite{hummel2018hybrid} further took the externality into account, %\emph{i.e.}, the number of ads displayed may vary with the bids of advertisers. 
and they aimed to guarantee that VCG bidders would bid truthfully, and GSP bidders could not obtain the same allocation at a lower price by misreporting their bids. This model for GSP bidders is different from our concept of VMs, and they also took the classes of bidders as public information.

\section{Conclusion}
In this work, we have investigated mechanism design for mixed environments with both UMs and VMs. %We have proposed two mechanisms, namely MPU and MPR, to guarantee IC, IR, robustness, and (approximately) optimal LSW in the cases with public and private class information, respectively. 
This work sheds light on future studies on private ROI constraints and also leaves several open problems. The foremost one is to close the gap between the approximation ratio's lower and upper bound.
We conjecture that MPR is in some way the right mechanism for this problem, and it would be interesting to make this more formal. Moreover, we will generalize our proposed mechanism for bidders with various types of ROI constraints.

\section*{Acknowledgements}
This work was supported in part by National Key R\&D Program of China No. 2021YFF0900800, in part by China NSF grant No. 61972230, 62025204, 62072303, 61972254, 62132018, 61902248, and 91846205, in part by Shanghai Science and Technology fund 20PJ1407900, in part by NSFShandong No. ZR2021LZH006, in part by Shandong Provincial Major Scientific and Technological Innovation Project No. 2021CXGC010108, in part by Alibaba Group through Alibaba Innovation Research Program, and in part by Tencent Rhino Bird Key Research Project. The opinions, findings, conclusions, and recommendations expressed in this paper are those of the authors and do not necessarily reflect the views of the funding agencies or the government.

\bibliography{aaai23}
\end{document}